\documentclass[12pt,onecolumn,draftclsnofoot,journal]{IEEEtran}
\usepackage{graphicx}
\usepackage{float}
\usepackage{epstopdf}
\usepackage[cmex10]{amsmath}
\usepackage{array}
\usepackage{cite}
\usepackage{amssymb}
\usepackage{amsfonts}
\usepackage{amsmath}
\usepackage{stackrel}
\usepackage{booktabs,multirow}
\usepackage{arydshln}
\usepackage{slashbox}
\usepackage{amsthm}
\usepackage{listings}
\usepackage{algorithm}
\usepackage{algorithmicx}
\usepackage{algpseudocode}
\usepackage{threeparttable}

\newtheorem{theo}{Theorem}

\newfloat{routine}{htbp}{loa}
\floatname{routine}{Routine}

\makeatletter
\newcommand{\algmargin}{\the\ALG@thistlm}
\makeatother
\newlength{\forwidth}
\algdef{SE}[parFOR]{parFor}{EndparFor}[1]
  {\parbox[t]{\dimexpr\linewidth-\algmargin}{
     \hangindent\forwidth\strut\algorithmicfor\ #1\ \algorithmicdo\strut}}{\algorithmicend\ \algorithmicfor}
\algnewcommand{\parState}[1]{\State
  \parbox[t]{\dimexpr\linewidth-\algmargin}{\strut #1\strut}}

\newlength{\ifwidth}
\algdef{SE}[parIF]{parIf}{EndparIf}[1]
  {\parbox[t]{\dimexpr\linewidth-\algmargin}{
     \hangindent\ifwidth\strut\algorithmicif\ #1\ \algorithmicdo\strut}}{\algorithmicend\ \algorithmicif}

\begin{document}

\title{Hardness Results on Finding Leafless Elementary Trapping
Sets and Elementary Absorbing Sets of LDPC Codes}
\author{Ali Dehghan, and Amir H. Banihashemi,\IEEEmembership{ Senior Member, IEEE}}

\maketitle


\begin{abstract}
Leafless elementary trapping sets (LETSs) are known to be the problematic structures in the error floor region of low-density parity-check (LDPC) codes over the additive white Gaussian (AWGN) channel under iterative decoding algorithms. While problems involving the general category of trapping sets, and the subcategory of elementary trapping sets (ETSs), have been shown to be NP-hard, similar results for LETSs, which are a subset of  ETSs are not available. In this paper, we prove that, for a general LDPC code, finding a LETS of a given size $a$ with minimum number of odd-degree check nodes $b$ is NP-hard to approximate within any approximation factor. We also prove that finding the minimum size $a$ of a LETS with a given $b$ is NP-hard to approximate within any approximation factor. Similar results are proved for elementary absorbing sets, a popular subcategory of LETSs.

\begin{flushleft}
\end{flushleft}

\end{abstract}

\section{introduction}

The error floor of low-density parity-check (LDPC) codes under iterative decoding algorithms is attributed to certain combinatorial structures in the Tanner graph of the code, collectively referred to as {\em trapping sets}. A trapping set is often classified by its size $a$ and the number of odd-degree (unsatisfied) check nodes $b$ in its induced subgraph. In this case, the trapping set is said to belong to the $(a,b)$ {\em class}. The problematic trapping sets that cause the error floor depend not only on the Tanner graph of the code, but also on the channel model, quantization scheme and the decoding algorithm. For variable-regular LDPC codes (both random and structured) over the additive white Gaussian noise (AWGN) channel, the culprits are known to be a subcategory of trapping sets, called {\em leafless elementary trapping sets (LETSs)}~\cite{MR3252383},~\cite{hashemi2015new}. Leafless ETSs are also the majority of problematic structures in the error floor region of irregular LDPC codes over the AWGN channel~\cite{HB-irregular}. The term ``elementary'' indicates that all the check nodes in the induced subgraph of the trapping set have degree one or two, and the term ``leafless'' means that each variable node is connected to at least two satisfied check nodes, i.e., the {\em normal graph}~\cite{MR3252383} of the trapping set contains no leaf.

Elementary TSs have been reported in numerous publications, including the pioneering work
of Richardson~\cite{Richardson}, to be the most harmful of TSs, based on simulations. Recently, also,
theoretical results were presented in~\cite{hashemilower} that demonstrated the smallest size $a$ of non-elememtary
$(a,b)$ TSs (NETSs) is generally larger than the smallest size of ETSs with the same $b$ value. Considering
that for a given $b$, TSs with smaller size $a$ are generally more harmful, the result of~\cite{hashemilower} provided
a theoretical justification for why ETSs are the most harmful among all TSs. More evidence that ETSs are more harmful than NETSs was most recently provided in~\cite{lastyoones}, where the authors examined a large number of LDPC codes and demonstrated, using exhaustive search, that dominant classes of trapping sets are those of ETSs. Moreover, it has been
observed that for many LDPC codes, the dominant ETSs in the error floor are those caused by
a combination of multiple cycles, where each variable node is part of at least one cycle, i.e., the
dominant ETSs are leafless, see, e.g.,~\cite{case, ZLT-TCOM-2010, karimi2012efficient}, and the references therein. This motivates our investigation into the computational complexity of finding LETSs in this paper.

For a given LDPC code, the knowledge of TSs, in general, and that of LETSs, in particular, is useful in estimating the error floor~\cite{cole,TB-2014}, devising decoding algorithms~\cite{kyung2012finding}, or designing codes~\cite{RR2, RR3} with low error floors. In any such application, one would need to find a list of dominant TSs which are the main contributors to the error floor. While the topic of relative harmfulness of different trapping sets with relation to the channel model, the decoding algorithm, and the quantization is still not fully understood, it is generally accepted that TSs with smaller values of $a$ and $b$ are more harmful. Counterexamples, however, exist where trapping sets with larger $a$ or $b$ values are more harmful than those with smaller values for these parameters~\cite{MR2292873}. There are also many examples where the relative harmfulness of trapping sets in different classes can change depending on the decoding algorithm, the channel model or the quantization scheme. It is thus of interest to devise algorithms that can find TSs in a wide range of $a$ and $b$ values. Motivated by this, there has been a flurry of research activity to characterize different categories of trapping sets and to devise fast search algorithms to find them~\cite{MR3252383},~\cite{hashemi2015new},~\cite{HB-irregular},~\cite{lastyoones},~\cite{kyung2012finding},~\cite{karimi2012efficient},~\cite{MR2951331},~\cite{hashemi2015characterization},~\cite{falsafain2016exhaustive}.
In particular, recently, Hashemi and Banihashemi~\cite{hashemi2015new} proposed a characterization of LETSs for variable-regular Tanner graphs based on a hierarchy of graphical structures that starts from a simple cycle and expands recursively to reach the targeted LETS. The characterization, referred to as {\em dpl}, is based on three simple expansion techniques, dubbed {\em dot (degree-one-tree), path}, and {\em lollipop}, and has the property that it generates a LETS at each and every step of the expansion process. Within this framework, the authors of~\cite{hashemi2015new} proved the optimality of {\em dpl} search, in that, it can find LETSs within a certain range of $a \leq a_{max}$ and $b \leq b_{max}$, exhaustively, by starting from simple cycles whose maximum length is minimized, and by generating the minimum number of undesirable structures that are outside the range of interest. The complexity of the $dpl$ search of \cite{hashemi2015new} depends highly on the multiplicity of short simple cycles in the graph and the number of LETS structures in different classes within the range of the search. It was shown in~\cite{dehghan2016new} that the multiplicities of simple cycles of different fixed lengths in random bipartite graphs tend to independent Poisson random variables with fixed expected values, as the size of the graph tends to infinity (for fixed degree distributions). It was further shown in~\cite{IT2cycle} that the average multiplicity of LETS structures within different $(a,b)$ classes, where $a$ and $b$ are fixed values, tends to zero or a non-zero constant asymptotically. 
With regards to the overall complexity of the $dpl$ search algorithm, it was proved in~\cite{HB-irregular} that,
assuming fixed values of $a_{max}$ and $b_{max}$ and a fixed maximum node degree and girth for the Tanner graph,
the worst-case complexity increases linearly with the code's block length $n$. It was further shown
in~\cite{HB-irregular} that the average complexity of the $dpl$ search, excluding that of the search for the input
simple cycles, is constant in $n$. While these results verify the efficiency of the $dpl$ search for
small (fixed) values of $a_{max}$ and $b_{max}$, they do not provide much insight into the computational complexity of the algorithm as a function of these parameters. In addition, it is of interest to know the inherent difficulty of finding LETSs regardless of the algorithm used.  In particular, we are interested in whether it is possible to find a polynomial time algorithm that can find LETSs of LDPC codes. The main contribution of this paper is to provide a negative answer to this question, in general. We recall that, for a given LDPC code, in applications involving error floor, one would be interested in a list of TSs with most contributions to the error floor. Such TSs are often those whose $a$ and $b$ values are relatively small, and as pointed out, e.g., in~\cite{MR2292873}, they may also include TSs whose parameters are proportional to the code's block length. For such cases, our results imply that, unless P = NP,  there is no polynomial time algorithm that can find all the LETSs of interest.

McGregor and Milenkovic \cite{mcgregor2010hardness} studied the computational complexity of finding different categories of trapping sets. In particular, they showed that for a given $a$, finding an $(a,b)$ trapping set with minimum $b$ is an NP-hard problem. Also, given $b$, it was shown that finding an $(a,b)$ ETS with minimum $a$ is NP-hard~\cite{mcgregor2010hardness}. In fact, the results of~\cite{mcgregor2010hardness} indicate that, for any constant $\epsilon$, 
there is no polynomial-time $\epsilon$-approximation algorithm for any of the two problems, unless RP = NP or P = NP, respectively.
On the other hand, when a problem is NP-hard, the restricted cases of that problem can be NP-hard or polynomial time solvable. We thus cannot judge the computational complexity of finding LETSs based on that of finding ETSs. Moreover, to the best of our knowledge, there is no result on the worst-case computational complexity of finding LETSs. The fact that LETSs are the main problematic structures in the error floor region of LDPC codes over the AWGN channel~\cite{case, ZLT-TCOM-2010, karimi2012efficient, butler2014error, hashemi2015new}, however, makes such a problem worth investigating.

In this paper, we study the computational complexity of finding LETSs in variable-regular LDPC codes. In particular,
we consider the following two problems: $(i)$ for a given  $(d_v, d_c)$-regular Tanner graph $G$ and an integer $a$, find an $(a,b)$ LETS with minimum $b$ in $G$, and
$(ii)$ for a given $b$, find a LETS with minimum size $a$ in $G$. We prove that both problems are NP-hard, even to approximate within any approximation factor, no matter how fast the approximation factor scales with the size of the problem. This implies that there is no polynomial-time algorithm to even approximate the solution to any of these problems, unless P = NP. 

In addition, we study the computational complexity of finding {\em elementary absorbing sets (EABS)}, a subcategory of LETSs. Absorbing sets are the fixed points of bit-flipping decoding algorithms~\cite{XB-2007},~\cite{DZAWN-2010}, and are also shown to be relevant in the context of quantized decoders over the AWGN channel~\cite{DZAWN-2010}. Among absorbing sets, the elementary ones are known to be the most problematic. We note that EABSs and LETSs are identical sets for variable-regular LDPC codes with $d_v=3$.
For other values of $d_v$, EABSs are a subset of LETSs. For EABSs, we prove that both Problems $(i)$ and $(ii)$ are still NP-hard, even to approximate within any approximation factor.

The organization of the rest of the paper is as follows: In Section~\ref{sec2}, we present some preliminaries including some definitions and notations. This is followed in Sections~\ref{sec3} and \ref{sec4} by our results on the worst-case computational complexity of finding LETSs and EABSs, respectively.
The paper is concluded with some remarks in Section~\ref{sec5}.

\section{Preliminaries}
\label{sec2}

We say that a graph $G$ is {\it simple} if it has no loop or parallel edges. Throughout this paper, all graphs are simple.
For a graph $G$, we denote the node set and the edge set of $G$ by $V(G)$ and $E(G)$, or by $V$ and $E$ (if there is no ambiguity about the graph), respectively.
The number of edges connected to a node $v$ is called the {\em degree} of $v$, and is denoted by $d(v)$.
Also, the {\it maximum degree} and the {\it minimum degree} of a graph $G$, denoted by $\Delta(G)$ and $\delta(G)$, respectively, are defined to be the maximum and the minimum degree of its nodes, respectively.
A node $v$ is called {\it leaf} if $d(v) = 1$. A {\it leafless} graph is a  graph $G$ with $\delta(G) \geq  2$.
In this work, we  consider graphs to be connected (note that disconnected graphs can be considered as the union of connected ones).

A graph $G=(V,E)$ is called {\it bipartite} if the node set $V$ can be
partitioned into two disjoint subsets $U$ and $W$, i.e., $V = U \cup W \text{ and } U \cap W =\emptyset $, such that every edge in $E$ connects a node from $U$ to a node from $W$. Such a bipartite graph is denoted by $G=(U \cup W ,E)$.

Each $m \times n$ parity check matrix $H$ of a linear block code of block length $n$, in general, and an LDPC code, in particular, can be represented by a bipartite (Tanner) graph $G =  (U \cup W,E)$, where $U = \{u_1, u_2, \ldots , u_n\}$ is the set of {\em variable nodes} and $W = \{w_1, w_2, \ldots , w_m \}$ is the set of {\em check nodes}.
A Tanner graph is called {\it variable-regular} with variable degree $d_v$ if the degree of every variable node is $d_v$. Also, a $(d_v, d_c)$-regular Tanner graph is a variable-regular Tanner graph in which the degree of every check node is $ d_c$.
For an LDPC code (Tanner graph),  if $H$ is full-rank, the {\em rate} is given by $R = 1 - \bar{d_v} / \bar{d_c}$, where $\bar{d_v}$ and $\bar{d_c}$ are the average variable and check degrees, respectively. In this work, we consider LDPC codes (Tanner graphs) in which $\bar{d_v}$ is constant with respect to $n$. This implies that, given $R$, in the asymptotic regime where $n \rightarrow \infty$, the density of $H$ and that of the Tanner graph tends to zero. This is consistent with the term ``low-density'' in ``low-density parity-check.''

Let $G=  (U \cup W,E)$ be a bipartite graph. For a set $S$,   where $S\subseteq U$, the set $N(S)$,  where $N(S)\subseteq W$, denotes the set of neighbors of $S$ in $G$. Also, the induced subgraph $G(S)$ of $S$ in $G$, is defined as the graph with the set of nodes $S \cup N(S)$ and the set of edges $\{u_iw_j: u_iw_j \in E , u_i\in S , w_j \in N(S)\}$.
In the graph $G(S)$, the set of check nodes with odd and even degrees  are denoted by $N_o(S)$ and $N_e(S)$, respectively. Also, the terms {\it unsatisfied check nodes} and {\it satisfied check nodes} are used to refer to the check nodes in $N_o(S)$ and $N_e(S)$, respectively.

For a given Tanner graph $G=  (U \cup W,E)$, a set $S \subset U$ is called an {\it $(a,b)$ trapping set (TS)}
if $|S| = a$ and $|N_o(S)| = b$. An {\it elementary trapping set (ETS)} is a trapping set for which each check node  in $G(S)$ has degree either one or two. For a given set $S$,  $S \subset U$,  we say that $S$ is a {\it leafless ETS (LETS)} if $S$ is an ETS and if the  graph which is obtained from $G(S)$  by removing all the check nodes of
degree one and their incident edges is leafless.
An {\em absorbing set (ABS)} $S$ is a TS for which all the variable nodes in $S$
are connected to more nodes in $N_e(S)$ than in $N_o(S)$.
Also, an {\it elementary absorbing set (EABS)} $S$ is an ABS for which all the check nodes in $G(S)$ have degree either one or two.

In computational complexity theory, the class P problems are those that can be solved in polynomial time. On the other hand, the class NP (non-deterministic polynomial-time) problems contains all decision problems for which the instances where the answer is ``yes'' have proofs that are verifiable by deterministic computations that can be performed in polynomial time. More formally, NP is the set of decision problems solvable in polynomial time by a theoretical non-deterministic Turing machine. A problem ${\cal P}$ is {\it NP-hard} when every problem ${\cal L}$ in class NP can be reduced in polynomial time to ${\cal P}$, i.e., assuming a solution for ${\cal P}$ takes one unit time, we can use ${\cal P}$'s solution to solve ${\cal L}$ in polynomial time. The complexity class P is contained in NP, but NP contains many more problems. The hardest problems in NP,  whose solutions are sufficient to deal with any other NP problem in polynomial time, are called {\it NP-complete}. The most important open question in complexity theory is whether P = NP. It is widely believed that the answer is negative. Polynomial-time reductions are frequently used in complexity theory to prove NP-completeness, i.e., if Problem ${\cal P}$ can be reduced to Problem ${\cal P}'$ in polynomial time, then  ${\cal P}$ is no more difficult than ${\cal P}'$, because whenever an efficient algorithm exists for ${\cal P}'$, one exists for ${\cal P}$ as well. Thus, if ${\cal P}$ is NP-complete, so is ${\cal P}'$.

Although many optimization problems cannot be solved in polynomial time unless P = NP, in many of these problems the optimal solution can be efficiently approximated to a certain degree. In this context, approximation algorithms are polynomial time algorithms that find approximate solutions to NP-hard problems with provable guarantees on the distance of the returned solution to the optimal one. In majority of the cases, the guarantee of such algorithms is a multiplicative one expressed as an approximation ratio or approximation factor, i.e., the optimal solution is guaranteed to be within a multiplicative factor of the returned solution. Some optimization problems, however, are NP-hard even to approximate to within a given approximation factor. In this work, we show that all the problems of interest in this work belong to this category of problems.

To prove the NP-hardness of the problems related to finding LETSs and EABSs, we use some existing results in Boolean logic on the NP-completeness of satisfiability problems. In the following, we first present some definitions and notations in Boolean logic. This is followed by the description of the Boolean satisfiability problems of interest.
 
Consider a formula $\Phi=(X,C)$ in Boolean logic, where the two sets $X=\{x_1,\ldots, x_n\}$ and $C=\{c_1,\ldots, c_m\}$ are the sets of variables and clauses of $\Phi$, respectively. Each variable can take one of the two {\em truth values} ``True ($1$ or $T$)'' or ``False ($0$ or $F$).'' A {\em truth assignment} $\ell(x),  x \in X$, for $\Phi=(X,C)$ is an assignment of truth values to all the variables in $X$, with corresponding assignment of a truth value to $\Phi$. We use the notation $\bar{\ell}$ for a truth assignment whose truth values $\bar{\ell}(x)$ are $T$ or $F$ if and only if the corresponding truth values $\ell(x)$ of $\ell$ are $F$ or $T$, respectively. We say that a formula $\Phi$ is in {\it conjunctive normal form (CNF)} if it is a conjunction of clauses, where a clause is a disjunction of literals. A {\it literal} is either a variable $x$ or the negation $\neg x$ of a variable $x$. For example, $(x\vee \neg y \vee z) \wedge (x\vee  y \vee z) \wedge (x\vee  \neg y \vee \neg z) $ is a CNF formula with the set of variables $\{x,y,z\}$ and the set of clauses $\{(x\vee \neg y \vee z),(x\vee  y \vee z),( x\vee  \neg y \vee \neg z)\}$.
Throughout this paper, all the formulas are assumed to be in conjunctive normal form. In the following, we sometimes represent a CNF formula by the collection of its clauses. With a slight abuse of notation, and to simplify the equations,
we would also refer to a CNF formula $\Phi=(X,C)$ to mean the set of its clauses $C$.

For a given formula $\Phi$, we say that $\Phi$ has a {\em $\gamma $--IN--$ \beta$ truth assignment} if each clause in $\Phi$ has exactly $ \beta$ literals and there is a truth assignment for $\Phi$ such that each clause  has exactly
$\gamma$ true literals. For instance, for the formula $(x\vee  y \vee z) \wedge (\neg x\vee  \neg y \vee z) \wedge (\neg x\vee  y \vee \neg z) $, the assignment $\ell:
\ell(x)= T, \:\ell(y)=\ell(z)=F$, is a $1$--IN--$3$  truth assignment.
A {\em $\gamma$--IN--$\beta$ SAT problem} is the problem of determining whether there exists a $\gamma $--IN--$ \beta$ truth assignment for a given formula $\Phi$ (where in $\Phi$ every clause contains $\beta$ literals).
We say that a formula is {\it monotone} if there is no negation in the formula.

The following problems are of interest in proving our results:
\begin{enumerate}
\item {\em Monotone $1$--IN--$3$ SAT}: Given a monotone formula $\Phi=(X,C)$ such that every clause in $C$ contains  three variables, is there a $1$--IN--$3$ truth assignment for $\Phi$?

\item  {\em  Cubic Monotone $1$-IN-$3$ SAT}: Given a monotone formula $\Phi=(X,C)$ such that every clause in $C$ contains  three variables and every variable  appears in exactly three clauses,
is there a $1$--IN--$3$ truth assignment for $\Phi$?

\item  {\em   Monotone $2$--IN--$\beta$ SAT}: Given a monotone formula $\Upsilon=(X,C)$  such that every clause in $C$ contains $\beta$ variables, is there a $2$--IN--$\beta$ truth assignment for $\Upsilon$?

\item {\em  Cubic Monotone $2$-IN-$\beta$ SAT}: Given a monotone formula $\Psi=(X,C)$
 such that every clause in $C$ contains $\beta$ variables and every variable  appears in exactly three clauses,
 is there a $2$--IN--$\beta$ truth assignment for $\Psi$?

\item  {\em  $\alpha$--Monotone $2$--IN--$\beta$ SAT}: Given a monotone formula $\phi=(X,C)$
 such that every clause in $C$ contains $\beta$ variables and every variable
 appears in exactly $\alpha$ clauses, is there a $2$--IN--$\beta$ truth assignment for $\phi$?

\end{enumerate}

Problems $1$ and $2$ above are shown in \cite{schaefer1978complexity} and~\cite{MR1863810}, respectively, to be NP-complete. As an intermediate result to prove the NP-hardness of the problems under consideration, we start from Problem $1$ and demonstrate in three steps, through polynomial-time reductions, that Problems $3$--$5$ are all NP-complete as well.

The problems of direct interest in this work are the followings:

\begin{itemize}
\item {\em Min-$b$-LETS}: Given an $(\alpha,\beta)$-regular Tanner graph $G$ and a positive integer $a$,
find the minimum non-negative integer $b$ such that there is an $(a,b)$ LETS in $G$.

\item {\em Min-$a$-LETS}: Given a Tanner graph $G$ and a non-negative integer $b$, find the minimum positive integer $a$ such that there is an $(a,b)$ LETS in $G$.

\item {\em Min-$b$-EABS}: Given a $(\alpha,\beta)$-regular Tanner graph $G$ and a positive integer $a$, find the minimum non-negative integer $b$ such that there is an $(a,b)$ EABS in $G$.

\item {\em Min-$a$-EABS}: Given a Tanner graph $G$ and a non-negative integer $b$, find the minimum positive integer $a$ such that there is an $(a,b)$ EABS in $G$.

\end{itemize}

In the rest of the paper, we prove that all the above problems are NP-hard to approximate within any approximation factor, i.e., there is no polynomial-time algorithm to approximate the solution to any of these problems, unless P = NP.

\section{Computational Complexity of Finding LETSs}
\label{sec3}

\begin{theo}\label{TH1}
For any integers $\alpha$ and $\beta$ satisfying $3 \leq \alpha \leq \beta$, {\em Min-$b$-LETS} is NP-hard to approximate within any approximation factor.
\end{theo}

\begin{proof}{
To prove the result, we reduce {\em Monotone $1$--IN--$3$ SAT} to a decision problem corresponding to {\em Min-$b$-LETS} in four steps.

{\bf Step 1.} (Reduction of {\em Monotone $1$--IN--$3$ SAT} to {\em Monotone $2$--IN--$\beta$ SAT})
Let $\Phi$ be an instance of {\em Monotone $1$--IN--$3$ SAT}. In this step, we convert the formula $\Phi$ into a monotone formula $\Upsilon$ such that in $\Upsilon$ every clause contains $\beta$ variables, and there is a truth assignment for variables in $\Phi$ such that  each clause in $\Phi$ has exactly
one true literal if and only if there is a truth assignment for variables in $\Upsilon$ such that  each clause in $\Upsilon$ has exactly two true literals.

We consider three cases: $(1) \beta=3$, $(2) \beta=4$, and $(3) \beta \geq 5$.

{\em Case 1 ($\beta=3$)}: Assume that $\ell$ is a $1$--IN--$3$ truth assignment for $\Phi$.
It is easy to see that $\bar{\ell}$ is a $2$--IN--$3$ truth assignment for $\Phi$. Hence, $\Phi$ has a $1$--IN--$3$ truth  assignment if and only if it has a $2$--IN--$3$ truth assignment. We thus choose $\Upsilon = \Phi$.

{\em Case 2 ($\beta=4$)}: Consider the formula $\Phi=(X,C)$ and let $x'$ be a new variable such that $x' \notin X$. Replace each clause $c$ in  the formula $\Phi$ with
$c \vee x'$, and call the resultant formula $\Upsilon=(X \cup x', C')$. 
If $\Phi$ has a $1$--IN--$3$ truth assignment $\ell$, then we select a truth assignment $\ell'$ for $\Upsilon$ by defining $\ell'(x) = \ell(x), \: \forall x \in X$, and $\ell'(x')=T$. It is now easy to see that $\ell'$ is a $2$--IN--$\beta$ truth assignment for $\Upsilon$. On the other hand, assume that $\ell'$ is a $2$--IN--$\beta$ truth assignment for $\Upsilon$. If $\ell'(x')=T$, then $\ell'$ is also a $1$--IN--$3$ truth assignment for $\Phi$. Otherwise, if $\ell'(x')=F$, then $\ell'$
is a $2$--IN--$3$ truth  assignment for $\Phi$. In this case, $\bar{\ell'}$ will be a $1$--IN--$3$ truth  assignment for $\Phi$.

{\em Case 3 ($\beta \geq 5$)}: Consider the collection of following $\beta-1$ clauses and call them $\mathcal{S}$:
\\
$(x_1 \vee x_2 \vee \ldots \vee x_{\beta-1}\vee y_1)$,\\
$(x_1 \vee x_2 \vee \ldots \vee x_{\beta-1}\vee y_2)$,\\
$ \vdots $\\
$(x_1 \vee x_2 \vee \ldots \vee x_{\beta-1} \vee y_{\beta-2})$,\\
$(y_1 \vee y_2 \vee \ldots \vee y_{\beta-2} \vee x_1 \vee x_2)$.

Assume that $\ell$ is an arbitrary  $2$--IN--$\beta$ truth assignment for $\mathcal{S}$, and let $i$ and $j$ be two arbitrary integers such that $1 \leq i<j \leq \beta-2$. If we have $\ell(y_i)\neq \ell(y_j)$, then,
based on Assignment $\ell$, it is impossible that both clauses $(x_1 \vee x_2 \vee \ldots \vee x_{\beta-1}\vee y_i)$ and $(x_1 \vee x_2 \vee \ldots \vee x_{\beta-1}\vee y_j)$ have exactly two true literals.
Thus, we must have $\ell(y_1)=\ell(y_2)=\cdots=\ell(y_{\beta-2})$. Since, based on Assignment $\ell$, there must be exactly two true literals in the clause $(y_1 \vee y_2 \vee \ldots \vee y_{\beta-2} \vee x_1 \vee x_2)$, and since $\beta-2 \geq 3$, we conclude  $\ell(y_1)=\ell(y_2)=\cdots=\ell(y_{\beta-2})=F$, and thus $\ell(x_1)=\ell(x_2)=T$. Consequently, $\ell(x_3)=\ell(x_4)= \cdots = \ell(x_{\beta-1})=F$.

For each $i$, $1\leq i \leq \beta-3$, consider the collection of following clauses and call them $\mathcal{S}_i$:
\\
$(x_1^i\vee x_2^i\vee \ldots\vee x_{\beta-1}^i\vee y_1^i)$,\\
$(x_1^i\vee x_2^i\vee \ldots\vee x_{\beta-1}^i\vee y_2^i)$,\\
$ \vdots $\\
$(x_1^i\vee x_2^i\vee \ldots\vee x_{\beta-1}^i\vee y_{\beta-2}^i)$,\\
$(y_1^i\vee y_2^i\vee \ldots\vee y_{\beta-2}^i\vee x_1^i\vee x_2^i)$.
\\
Now, consider the collection of clauses $\mathcal{H}= C \cup \{\cup_{i=1}^{\beta-3} \mathcal{S}_i \}$, where $C$ is the set of clauses in $\Phi$. In $\mathcal{H}$,
replace each clause $c \in C$ (which has $3$ variables) with the clause $c \vee x_3^1 \vee x_3^2 \vee \ldots \vee x_3^{\beta-3}$.
Call the resultant formula $\Upsilon$. In $\Upsilon$, each clause has exactly $\beta$ variables.
For each $i$, $1\leq i \leq \beta-3$, if $\ell$ is a $2$--IN--$\beta$ truth assignment for $\mathcal{S}_i$, then $\ell(x_3^i)=F$.  Thus, the formula $\Phi$ has a $2$--IN--$3$ truth assignment if and only if
$\Upsilon$ has a $2$--IN--$\beta$ truth assignment. On the other hand, the formula $\Phi$ has a $2$--IN--$3$ truth  assignment if and only if it has a $1$--IN--$3$ truth assignment. This completes the proof for this case.

{\bf Step 2.} (Reduction of {\em Monotone $2$--IN--$\beta$ SAT} to {\em  Cubic Monotone $2$--IN--$\beta$ SAT})
 In this step, we convert $\Upsilon$ to a monotone formula $\Psi$ such that in $\Psi$ every clause contains $\beta$ variables, every variable
 appears in exactly $3$ clauses, and there is a truth assignment for $\Upsilon$ such that
 each clause in $\Upsilon$ has exactly two true literals if and only if there is a truth assignment for $\Psi$ such that
 each clause in $\Psi$ has exactly two true literals.

\begin{figure}[ht]
\begin{center}
\includegraphics[scale=.5]{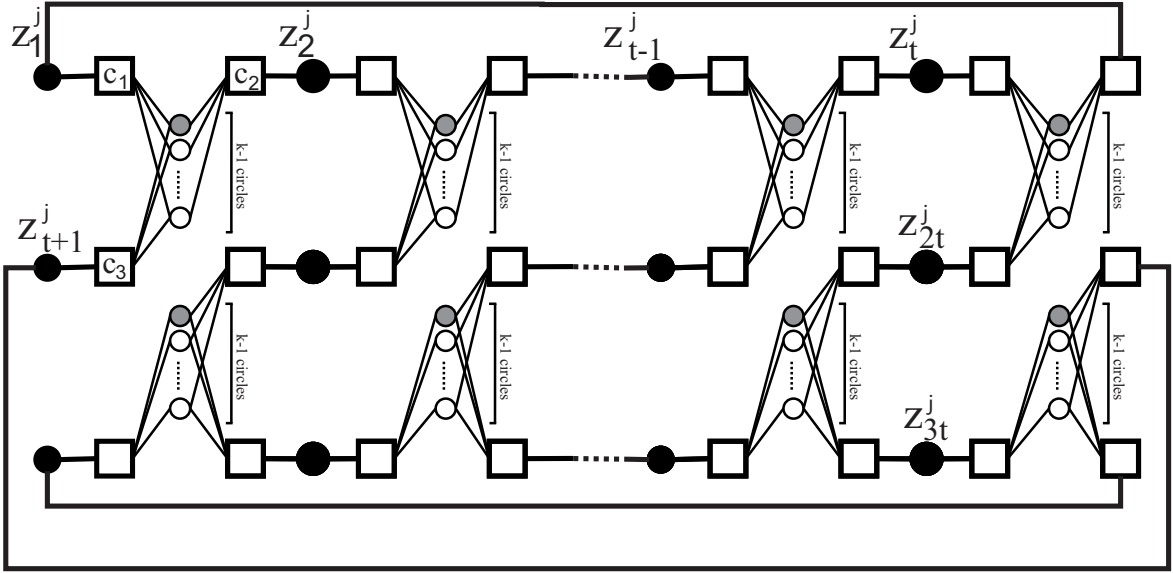}
\caption{Graph $F_{t,k,j}$ representing Formula $\Omega_{t,k,j}$.
} \label{graphA1}
\end{center}
\end{figure}

First, consider the graph $F_{t,k,j}$ shown in Figure \ref{graphA1}. This graph represents a formula in the following way: Every circle is a variable, every square is a clause and a clause $c$ contains a variable $v$ if and only if they are connected through an edge in the graph. Call the formula that is represented by $F_{t,k,j}$, $\Omega_{t,k,j}$. In $\Omega_{t,k,j}$, every variable occurs at most three times, each clause contains exactly $k \geq 3$ variables and there
is no negation. Parameter $j$ is an index that all the variables in $\Omega_{t,k,j}$  have as superscript. The graph $F_{t,k,j}$ contains $3t$ black variable nodes, each of degree two.

It is easy to see that there are some $2$--IN--$k$ truth assignments for $\Omega_{t,k,j}$. For instance, one can assign $T$ to black and grey variable nodes and $F$ to white ones.
Another example is to assign $F$ to black variable nodes, $T$ to grey and one of the white variable nodes in each group of $k-2$ white nodes, and $F$ to the rest of white nodes. In the following, we prove that in any $2$-IN-$k$ truth assignment for $\Omega_{t,k,j}$, the values of all black variable nodes must be equal, i.e., either they are all assigned ``$T$,'' or they are all equal to ``$F$'' ({\bf Fact 1}).
In $F_{t,k,j}$, consider the three black nodes $z_1^j$, $z_2^j$ and $z_{t+1}^j$. Clause $c_1$ contains $z_1^j$, Clause $c_2$ contains $z_2^j$ and Clause $c_3$ contains $z_{t+1}^j$.
Also, there are $k-1$ variables, which we call them $v_1^j,v_2^j, \ldots , v_{k-1}^j$, such that $c_1=(z_1^j  \vee v_1^j \vee \ldots \vee v_{k-1}^j)$, $c_2=(z_2^j  \vee v_1^j \vee \ldots \vee v_{k-1}^j)$ and $c_3=(z_{t+1}^j \vee v_1^j \vee \ldots \vee v_{k-1}^j)$.
Assume that $\ell$ is a  $2$--IN--$k$ truth assignment for $\Omega_{t,k,j}$. Then, we must have $\ell(z_1^j)=\ell(z_2^j)=\ell(z_{t+1}^j)$. Otherwise, it is impossible to have exactly two true literals in all three clauses $c_1$, $c_2$ and $c_3$.
Using a similar argument, one can see that in any $2$-IN-$k$ truth assignment, the values of all the black variable nodes must be equal

Now, consider the formula $\Upsilon'=(X,C')$, which has the same set of variables as $\Upsilon=(X,C)$, but each clause in $C$ is copied nine times to form $C'$.
Each variable $x_j$ in $\Upsilon'$, therefore, appears in $3h(x_j)$ clauses, where $h(x_j)\geq 3$. Now, for each variable $x_j$ in $\Upsilon'=(X,C')$,
consider a copy of the formula  $\Omega_{h(x_j),\beta,j}$. In this formula (corresponding to the graph $F_{h(x_j),\beta,j}$), let $\{z_1^j,z_2^j, \ldots, z_{3h(x_j)}^j\}$ be the set of black variable nodes.
Replace each appearance of $x_j$  in $\Upsilon'$ with one of the variables $z_1^j,z_2^j, \ldots, z_{3h(x_j)}^j$, and call the resultant formula $\Upsilon''$.
Now, consider the set of clauses $\{\bigcup_{x_j\in X} \Omega_{h(x_j),\beta,j} \cup \Upsilon'' \}$, and call it $\Psi$. In $\Psi$, every clause contains $\beta$ variables and  every variable
appears in exactly three clauses. By Fact 1, it is straightforward to verify that Formula $\Psi$ has a $2$--IN--$\beta$ truth assignment if and only if $\Upsilon$ has a $2$--IN--$\beta$ truth assignment.

{\bf Step 3.} (Reduction of {\em  Cubic Monotone $2$--IN--$\beta$ SAT} to {\em $\alpha$--Monotone $2$--IN--$\beta$ SAT})
In this step, we convert $\Psi$ to a monotone formula $\phi$ such that in $\phi$ every clause contains $\beta$ variables, every variable appears in $\alpha$ clauses, and $\Psi$ has a $2$--IN--$\beta$ truth assignment if and only if $\phi$ has a $2$--IN--$\beta$ truth assignment.
The proof for $\alpha = 3$ is trivial. In the following, we thus consider $\alpha > 3$.

For each variable $x$ in $\Psi=(X,C)$ and  each $i$, $1\leq i \leq \alpha$, consider the following clauses and call them
$\mathcal{D}_{i,x}$:
\\
$(x_i\vee y^1_{1,x} \vee y^1_{2,x}\ldots \vee y^1_{\beta-1,x})$,\\
$(x_i\vee y^2_{1,x} \vee y^2_{2,x}\ldots \vee y^2_{\beta-1,x})$,\\
$ \vdots $\\
$(x_i\vee y^{\alpha-3}_{1,x} \vee y^{\alpha-3}_{2,x}\ldots \vee y^{\alpha-3}_{\beta-1,x})$.

Let $T=\bigcup_{i,x} \mathcal{D}_{i,x}$. If $\Psi$ has $\gamma$ variables, then $T$ has $\alpha (\alpha-3)\gamma$ clauses. Also, note that for each $x\in X$ and each $1\leq i \leq \alpha$, the variable $x_i$ appears $\alpha-3$ times in $T$.

For each variable $x \in X$, the set $T$ contains the following clauses:
\\
$(x_1\vee y^1_{1,x} \vee y^1_{2,x}\ldots \vee y^1_{\beta-1,x})$,\\
$(x_2\vee y^1_{1,x} \vee y^1_{2,x}\ldots \vee y^1_{\beta-1,x})$,\\
$ \vdots $\\
$(x_{\alpha}\vee y^1_{1,x} \vee y^1_{2,x}\ldots \vee y^1_{\beta-1,x})$.\\
Therefore, if $\ell$ is a $2$-IN-$\beta$ truth assignment for $T$, then $\ell(x_1)=\ell(x_2)=\cdots=\ell(x_{\alpha})$ ({\bf Fact 2}).

Consider $\alpha$ copies of $\Psi=(X,C)$, and for each $x\in X$ and $i$, $ 1\leq i \leq \alpha$, in copy $i$, replace all the $x$ variables with $x_i$.
Call the resultant formula $\Psi'$. Next, consider the union of $\Psi'$ and the set of clauses $T$, and call it $\phi$. In $T$, $x_i$ appears $\alpha -3$ times and in $\Psi'$, it appears 3 times. Thus, in total (in $\phi$),  Variable $x_i$
appears $\alpha$ times.

Now, we show that $\Psi$ has a $2$--IN--$\beta$ truth assignment if and only if $\phi$ has a $2$--IN--$\beta$ truth assignment. Let $\ell$ be a $2$--IN--$\beta$ truth assignment for $\phi$.
Define the assignment $f$ for $\Psi$ such that $f(x)=T$ or $F$ if $\ell(x_1)=T$ or $F$, respectively. It is easy to see that $f$ is a $2$--IN--$\beta$ truth assignment for $\Psi$.
Now, assume that $f$ is $2$--IN--$\beta$ truth assignment for $\Psi$. Define the assignment $\ell$ for $\phi$ as following:

\begin{center}
$\ell(v)=
\begin{cases}
   f(x),       &\text{if  }\,\,v=x_i,\, x\in X,\, 1\leq i \leq \alpha\\
   T,       &\text{if  }\,\,v=y^j_{1,x},\, 1\leq j \leq \alpha-3,\, x\in X\\
   T,       &\text{if  }\,\,v=y^j_{2,x},\, f(x)=F,\, 1\leq j \leq \alpha-3,\,  x\in X\\
   F,      &\text{if  }\,\,v=y^j_{2,x},\, f(x)=T,\, 1\leq j \leq \alpha-3,\,  x\in X\\
   F,      &\text{Otherwise}.\
\end{cases}$
\end{center}

By Fact 2, it is easy to check that $ \ell$ is   a $2$--IN--$\beta$ truth assignment for $\phi$.

{\bf Step 4.} (Reduction of  {\em $\alpha$--Monotone $2$--IN--$\beta$ SAT} to {\em Min-$b$-LETS})
Let $\phi$ be an instance of {\em $\alpha$--Monotone $2$--IN--$\beta$ SAT} Problem. We first construct a bipartite graph $G =  (U \cup W,E)$ from $\phi$: For each variable $x$ and each clause $c$ in $\phi$, we create a variable node $x$ in $U$, and a check node $c$ in $W$, respectively. Next, if Variable $x$ appears in Clause $c$ of $\phi$, we connect the variable node $x$ to the check node $c$ in $G$. Clearly, the resultant graph is an $(\alpha,\beta)$-regular Tanner graph. We use $G$ as an input instance of {\em Min-$b$-LETS} Problem and select $a=\dfrac{2|W|}{\alpha}$ as the input size of the LETS.
Now, in what follows, we first show that if there is a $(\dfrac{2|W|}{\alpha},b)$ LETS in $G$, then we must have $b=0$. Next, we show that
there is a $(\dfrac{2|W|}{\alpha},0)$ LETS in $G$ if and only if $\phi$ has a $2$--IN--$\beta$ truth assignment. (We note that the existence of any polynomial-time algorithm that provides an approximate solution to
{\em Min-$b$-LETS} within any approximation factor can be used to solve the problem exactly for the input $a = \dfrac{2|W|}{\alpha}$, and thus, unless P = NP, such an algorithm does not exist.)

In the following, we show that if $G$ has a LETS $S$ of size $a=\dfrac{2|W|}{\alpha}$, then in the induced subgraph $G(S)$ of $S$ in $G$, we have $|N_e(S)|=|W|$ and $|N_o(S)|=b=0$ ({\bf Fact 3}).
Let $S $ be a LETS of size $a=\dfrac{2|W|}{\alpha}$ in $G$.
If we count the number of edges in $G(S)$ from the $U$ side of the graph, we  have
\begin{equation}\label{D01}
 \displaystyle|E(G(S))|= a\times  d_v= \dfrac{2|W|}{\alpha} \times \alpha= 2|W|\:.
\end{equation}
On the other hand, by counting the number of edges in $G(S)$ from the $W$ side, we have
\begin{equation}\label{E01}
\displaystyle|E(G(S))| \leq 2|N(S)| \leq 2W\:,
\end{equation}
where the first inequality is due to the fact that in a LETS, check degrees are either one or two. From (\ref{D01}) and (\ref{E01}), we obtain $|N(S)| = |N_e(S)|=|W|$, and $|N_o(S)|=0$.

Now,  we prove the ``only if'' part of the main claim. Assume that $G$ has a $(\dfrac{2|W|}{\alpha},0) $ LETS $S$.
Define the  truth assignment $\ell:X\rightarrow \{T, F\}$, such that $\ell(x)=T$ if $x\in S$, and $\ell(x)=F$, otherwise. By Fact 3, every check node in $G$ is connected to exactly two variable nodes in $S$ and thus every clause in $\phi$ has exactly two true literals. Assignment $\ell$ is thus a  $2$--IN--$\beta$ truth assignment for $\phi$.

For the ``if part'' of the claim, assume that $\phi$ has a $2$--IN--$\beta$ truth assignment $\ell$. Let $S $ be a  subset of $U$ containing all the variables nodes for which $\ell(x)=T$. Since $\ell$ is a $2$--IN--$\beta$ truth assignment for $\phi$,
we have $N(S) = N_e(S) = W$, and each check node in $N_e(S)$ has degree two. This means $S$ is a LETS with $b=0$.
Now, by counting the number of edges form the two sides of $G(S)$, we have $a=\dfrac{2|W|}{\alpha}$.
Hence, $G(S)$ is a $(\dfrac{2|W|}{\alpha},0)$ LETS in $G$.
}\end{proof}

\begin{theo}
{\em Min-$a$-LETS} is NP-hard to approximate within any approximation factor.
\end{theo}

\begin{proof}{
To prove the result, we reduce {\em Cubic Monotone $1$--IN--$3$ SAT} Problem to a decision problem corresponding to {\em Min-$a$-LETS} Problem.

Let $\Phi=(X,C)$ be an instance of {\em Cubic Monotone $1$--IN--$3$ SAT} Problem.
It is easy to see that $|X|=|C|$. Throughout the proof, we use $\eta$ to denote $|X|$. First, we construct a Tanner graph $G =  (U \cup W,E)$ from the formula $\Phi$:
For each variable $x \in X$, we create two variable nodes $x_1$ and $x_2$ in $U$. For each clause $c \in C$, we create a copy of the gadget $\mathcal{D}_{c,2\eta+1}$ shown in Fig. \ref{graphA2}. (In Fig. \ref{graphA2},
circles and squares represent variable and check nodes, respectively.)
Next, for each clause $c \in C$, if variable $x$ appears in $c$, then we connect the variable nodes $x_1$ and $x_2$ to the check nodes $w_c^1$ and $w_c^2$ of the gadget $\mathcal{D}_{c,2\eta+1}$, respectively.
(See Fig. \ref{graphA4} for an example of the Tanner graph $G$ constructed from a cubic monotone formula with three variables in each clause.)

\begin{figure}[ht]
\begin{center}
\includegraphics[scale=.7]{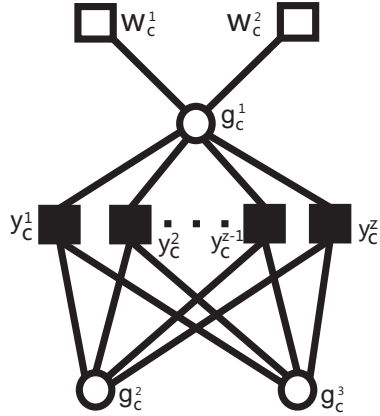}
\caption{Gadget $\mathcal{D}_{c,z}$.
} \label{graphA2}
\end{center}
\end{figure}

Note that Tanner graph $G$
has $5\eta$ variable nodes and $\eta (2\eta+3)$ check nodes. We use $G$ as the input instance of {\em Min-$a$-LETS} Problem and select $b = \eta (2\eta+1)$ as the input value of $b$.
Now, in the following, we first show that if $G$ has an $(a,\eta (2\eta+1))$ LETS, then, we must have $a=\eta‎+ ‎2\eta/3$ (Fact 10). Next, we prove that 
$G$ has a $(\eta‎+ ‎2\eta/3,\eta (2\eta+1))$ LETS if and only if $\Phi$ has a $1$--IN--$3$ truth assignment. (Any polynomial-time algorithm that can solve {\em Min-$a$-LETS} approximately within any approximation factor, 
can also provide an exact solution to the problem for the input $b=\eta (2\eta+1)$. Such an algorithm thus cannot exist, unless P = NP.) 

The followings are some properties of LETSs in $G$:

\begin{figure}[ht]
\begin{center}
\includegraphics[scale=.5]{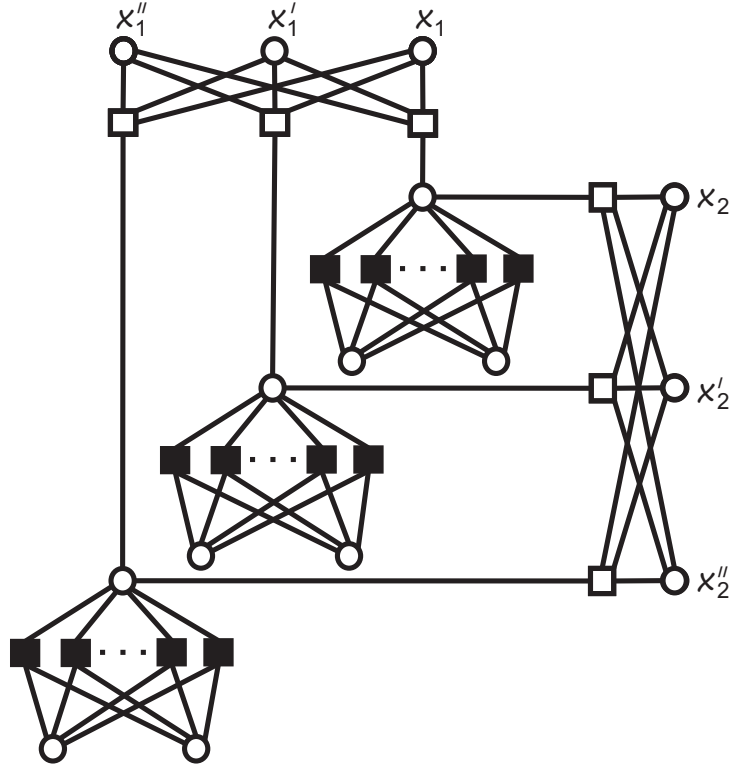}
\caption{Tanner graph $G$ corresponding to the formula
$(x\vee x' \vee x'') \wedge (x\vee x' \vee x'') \wedge (x\vee x' \vee x'') $.
} \label{graphA4}
\end{center}
\end{figure}

\begin{itemize}
\item {{\bf Fact 4}} Consider a LETS $S$ in $G$. Then, for each $c\in C$, since the neighbors of the nodes $y^1_c, y^2_c, \ldots, y^{2\eta+1}_c$ are the same,  we have $\{y^1_c, y^2_c, \ldots, y^{2\eta+1}_c\}\subset N_o(S)$ or $\{y^1_c, y^2_c, \ldots, y^{2\eta+1}_c\}\subset N_e(S)$ or none of the nodes $y^1_c, y^2_c, \ldots, y^{2\eta+1}_c$ is in $N(S)$.
\item {{\bf Fact 5}} For a LETS $S$ in $G$, if $\{y^1_c, y^2_c, \ldots, y^{2\eta+1}_c\}\subset N_o(S)$, then $ g^1_c \in S$, $ g^2_c\notin S$ and $ g^3_c\notin S$.
\item {{\bf Fact 6}}  Consider a LETS $S$ in $G$ with $b = \eta (2\eta+1)=2\eta^2+\eta$. Then, for every $c\in C$, $\{y^1_c, y^2_c, \ldots, y^{2\eta+1}_c\}\subset N_o(S)$. (Otherwise,  $b$ is at most $\eta (2\eta+3)-(2\eta+1)=2\eta^2+\eta-1$, which is a contradiction.)
\item {{\bf Fact 7}}  For a LETS $S$ in $G$ with $b=\eta (2\eta+1)$, for each $c\in C$, we have $ g^1_c\in S$,  $ g^2_c\notin S$ and $ g^3_c\notin S$. (Follows from Facts 5 and 6.)
\item {{\bf Fact 8}} For a LETS $S$ in $G$ with $b=\eta (2\eta+1)$, by Facts 6 and 7, for each $c\in C$, $w_c^1 \in N_e(S)$ and $w_c^2\in N_e(S)$. Thus, for each $c\in C$,
the check node $w_c^1$ has exactly one neighbor in the set $\{x_1: x\in X\}$.
\item {{\bf Fact 9}} Similar to Fact 8, for a LETS $S$ in $G$ with $b=\eta (2\eta+1)$, for each $c\in C$, the check node $w_c^2$ has exactly one neighbor in the set $\{x_2: x\in X\}$.
\item {{\bf Fact 10}}  For an $(a,b)$ LETS $S$ in $G$ with $b=\eta (2\eta+1)$, we must have $a=\eta+ 2\eta/3$. (This follows from Facts 7--9 and that the degree of each variable node in $\{x_1,x_2: x\in X\}$ is three.)
\end{itemize}

Now, we prove that $G$ has an $(\eta+ 2\eta/3,\eta (2\eta+1)) $ LETS  if and only if $\Phi$ has a $1$--IN--$3$ truth assignment.
Assume that $G$ has an $(\eta+ 2\eta/3,\eta (2\eta+1)) $ LETS $S$. Define the  truth assignment $\ell: X \rightarrow \{T,F\}$, such that for each $x\in X$, $\ell(x)=T$, if $x_1\in S$, and $\ell(x)=F$, otherwise.
By Fact 8, $\ell$ is a $1$--IN--$3$ truth assignment for $\Phi$.
Next, assume that the formula $\Phi$ has a $1$--IN--$3$ truth assignment $\ell$. Consider variables $x\in X$ with truth value $T$, and include their corresponding $x_1$ and $x_2$ variable nodes in $S$. Also, for each $c \in C$,  include $g^1_c$ in $S$. It is then easy to see that $S$ is an $(\eta+ 2\eta/3,\eta (2\eta+1)) $ LETS in $G$.
}\end{proof}

\section{Computational Complexity of Finding EABSs}
\label{sec4}

\begin{theo}\label{TH3}
For any integers $\alpha$ and $\beta$ satisfying $3 \leq \alpha \leq \beta$, {\em Min-$b$-EABS} is NP-hard to approximate within any approximation factor.
\end{theo}

\begin{proof}{
Similar to Step 4 of the proof of Theorem~\ref{TH1}, we reduce {\em $\alpha$--Monotone $2$--IN--$\beta$ SAT} Problem to {\em Min-$b$-EABS} Problem.
The construction of the Tanner graph $G$ is identical. In Step 4 of the proof of Theorem \ref{TH1}, we proved  that if there is a $(\dfrac{2|W|}{\alpha},b)$ LETS in $G$, then $b=0$, and also that there is a $(\dfrac{2|W|}{\alpha},0)$ LETS in $G$ if and only if the formula $\phi$ has a $2$--IN--$\beta$ truth assignment. Similarly, it can be proved that if there is a $(\dfrac{2|W|}{\alpha},b)$ EABS in $G$, then $b=0$.
Moreover, considering that in a variable-regular Tanner graph, LETSs with $b=0$ are identical to EABSs with $b=0$, we conclude that the necessary and sufficient condition for
$G$ to have a  $(\dfrac{2|W|}{\alpha},0)$ EABS is that the formula $\phi$ has a $2$--IN--$\beta$ truth assignment.
}\end{proof}

\begin{theo}\label{TH4}
{\em Min-$a$-EABS} is NP-hard to approximate within any approximation factor.
\end{theo}

\begin{proof}{
We reduce {\em Cubic Monotone $1$--IN--$3$ SAT} Problem to a decision problem corresponding to {\em Min-$a$-EABS} Problem.

Let $\Phi=(X,C)$ be an instance of {\em Cubic Monotone $1$--IN--$3$ SAT} Problem, where $\eta = |X| = |C|$.
In the following, we construct a Tanner graph $G =  (U \cup W,E)$ from $\Phi$. For each variable $x \in X$, we create two variable nodes $x_1$ and $x_2$ in $U$. For each clause $c \in C$, we create a copy $ \mathcal{F}_{c,2\eta+1}$ of the gadget  shown in Fig. \ref{graphA3}.
For every clause $c \in C$, and for every variable $x \in X$ that appears in $c$, we connect variable nodes $x_1$ and $x_2$ to the check nodes $w_c^1$ and $w_c^2$, respectively.

\begin{figure}[ht]
\begin{center}
\includegraphics[scale=.7]{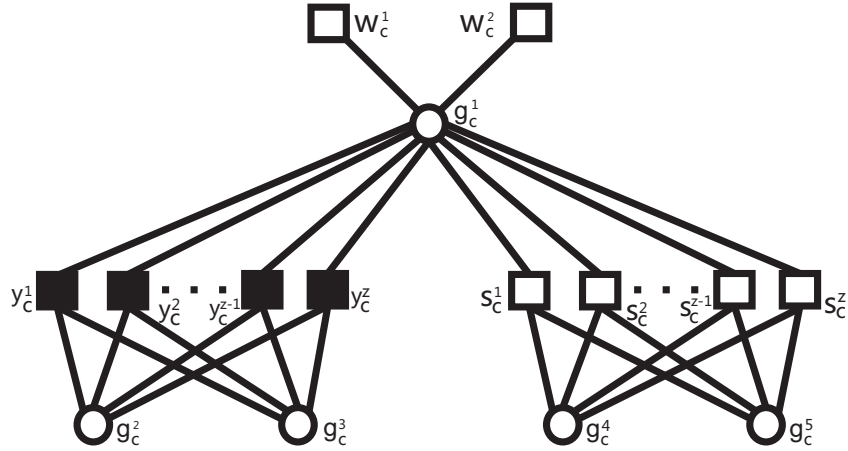}
\caption{Gadget $\mathcal{F}_{c,z}$.
} \label{graphA3}
\end{center}
\end{figure}

Note that Tanner graph $G$ has $7\eta$ variable nodes and $\eta (4\eta+4)$ check nodes.
We use $G$ as the input instance of {\em Min-$a$-EABS} Problem and select $b = \eta (2\eta+1)$ as the input value of $b$.
Now, in what follows, we first show that if $G$ has an $(a,\eta (2\eta+1))$ EABS, then we must have $a=2\eta‎+ ‎2\eta/3$ (Fact 18). Next, we prove that 
$G$ has a $(2\eta‎+ ‎2\eta/3,\eta (2\eta+1))$ EABS if and only if $\Phi$ has a $1$--IN--$3$ truth assignment. (Again, the fact that any approximate polynomial-time algorithm can find the exact solution to 
{\em Min-$a$-EABS} for $G$ and the input $b=\eta (2\eta+1)$ implies that the problem is NP-hard to approximate.)

In the following, we study some properties of EABSs in $G$:
\begin{itemize}
\item {{\bf Fact 11}} Consider an EABS $S$ in $G$. For each $c\in C$, we have either $\{y^1_c, y^2_c, \ldots, y^{2\eta+1}_c\}\subset N_o(S)$ or $\{y^1_c, y^2_c, \ldots, y^{2\eta+1}_c\}\subset N_e(S)$ or none of the nodes $y^1_c, y^2_c, \ldots, y^{2\eta+1}_c$ is in $N(S)$. Similarly, for each $c\in C$, we have $\{s^1_c, s^2_c, \ldots, s^{2\eta+1}_c\}\subset N_o(S)$ or $\{s^1_c, s^2_c, \ldots, s^{2\eta+1}_c\}\subset N_e(S)$ or none of the nodes $s^1_c, s^2_c, \ldots, s^{2\eta+1}_c$ is in $N(S)$.
\item {{\bf Fact 12}} For an EABS $S$ in $G$, if $\{y^1_c, y^2_c, \ldots, y^{2\eta+1}_c\}\subset N_o(S)$, then $ g^1_c\in S$, $ g^2_c\notin S$ and $ g^3_c\notin S$.
Similarly, if $\{s^1_c, s^2_c, \ldots, s^{2\eta+1}_c\}\subset N_o(S)$, then $ g^1_c\in S$, $ g^4_c\notin S$ and $ g^5_c\notin S$.
\item {{\bf Fact 13}} For an EABS $S$ in $G$, if $\{y^1_c, y^2_c, \ldots, y^{2\eta+1}_c \}\subset N_o(S)$, then $\{s^1_c, s^2_c, \ldots, s^{2\eta+1}_c\}\subset N_e(S)$. (Otherwise variable node $g^1_c$ is connected to
more nodes in $N_o(S)$ than in $N_e(S)$, which is in contradiction with the definition of absorbing sets.) Similarly, if $\{s^1_c, s^2_c, \ldots, s^{2\eta+1}_c\}\subset N_o(S)$, then $\{y^1_c, y^2_c, \ldots, y^{2\eta+1}_c\}\subset N_e(S)$.
\item {{\bf Fact 14}} Consider an EABS $S$ in $G$ with $b=\eta (2\eta+1)=2\eta^2+\eta$.  For each $c\in C$, $\{y^1_c, y^2_c, \ldots, y^{2\eta+1}_c\}\subset N_o(S)$ or  $\{s^1_c, s^2_c, \ldots, s^{2\eta+1}_c\}\subset N_o(S)$.
(Otherwise, by Fact 13, $b$ can be at most $2\eta^2+\eta-1$.)
\item {{\bf Fact 15}} Consider an EABS $S$ in $G$ with $b=\eta (2\eta+1)$. For each $c\in C$, $g^1_c\in S$, and exactly one of the variable nodes $g^2_c$, $g^3_c$, $g^4_c$, or $g^5_c$ is also in $S$. (Follows from Facts 13 and 14.)
\item {{\bf Fact 16}} For an EABS $S$ in $G$ with $b=\eta (2\eta+1)$, for each $c\in C$, $w_c^1,w_c^2\in N_e(S)$. (By Facts 14 and 15, and the definition of absorbing sets.)
\item {{\bf Fact 17}} For an EABS $S$ in $G$ with $b=\eta (2\eta+1)$, for each $c\in C$, the check node $w_c^1$ has exactly one neighbor in the set $\{x_1: x\in X\}$. Similarly,
for each $c\in C$, the check node $w_c^2$ has exactly one neighbor in the set $\{x_2: x\in X\}$.
\item {{\bf Fact 18}} For an $(a,b)$ EABS $S$ in $G$ with $b=\eta (2\eta+1)$, we have $a=2\eta+ 2\eta/3$. (Follows from Facts 15--17, and that the degree of each variable node in $\{x_1,x_2: x\in X\}$ is three.)
\end{itemize}

Now, we prove that $G$ has an $(\eta+ 2\eta/3,\eta (2\eta+1)) $ EABS  if and only if $\Phi$ has a $1$--IN--$3$ truth assignment.
Assume that $G$ has a $(2\eta+ 2\eta/3,\eta (2\eta+1)) $ EABS $S$. Define the truth assignment $\ell$ such that for each $x \in X$, $\ell(x)=T$ if $x_1\in S$, and $\ell(x)=F$, otherwise.
By Fact 17, $\ell$ is a $1$--IN--$3$ truth assignment for $\Phi$. Conversely, assume that $\Phi$ has a $1$--IN--$3$ truth assignment $\ell$. Let $S $ be a subset of $U$
such that for each $x \in X$, we have $x_1,x_2\in S$ if and only if $\ell(x)=T$. Also, for each $c\in C$, add $g^1_c$ and $ g^2_c$ to $S$. It is easy to see that $S$ is a  $(2\eta+ 2\eta/3,\eta (2\eta+1)) $ EABS in $G$.
}\end{proof}

\section{CONCLUSION}
\label{sec5}

In this paper, we discussed the computational complexity of finding leafless elementary trapping sets (LETSs) and elementary absorbing sets (EABSs) of LDPC codes,
and proved that such problems are NP-hard to even approximate within any approximation factor. This, under the assumption of $\text{P}\neq \text{NP}$, implies that there does not exist any polynomial-time algorithm to find
such structures (of even the smallest size). The hardness results proved in this paper for LETSs and EABSs are stronger than similar results proved in~\cite{mcgregor2010hardness} for trapping sets and ETSs, in the sense that,
while the results of~\cite{mcgregor2010hardness} indicate that there is no polynomial-time $\epsilon$-approximation algorithm, for any constant $\epsilon$, to solve the TS and ETS problems, 
our results imply that no polynomial-time approximation algorithm exists to solve the LETS and EABS problems no matter how fast the approximation factor of the algorithms scales with the size of the problem.

We also note that the intermediate results obtained in this work on NP-hardness of {\em   Monotone $2$--IN--$\beta$ SAT}, {\em  Cubic Monotone $2$-IN-$\beta$ SAT}, and  
{\em  $\alpha$--Monotone $2$--IN--$\beta$ SAT} Problems may be useful in the study of the hardness of other problems which are of interest in Computer Science or Coding.  

\bibliographystyle{ieeetr}

\end{document}